\documentclass[runningheads]{llncs}
\usepackage{geometry}
\usepackage{comment}
\usepackage{todonotes}
\usepackage{amsmath}
\usepackage{algorithm}
\usepackage{algpseudocode}
\usepackage{amsfonts}
\usepackage{multirow}
\usepackage[misc,geometry]{ifsym}

\usepackage{graphicx}
\begin{document}
\title{An improved bound for the price of anarchy for related machine scheduling}
%
\author{Andr{\'e} Berger\orcidID{0000-0002-6409-1963} \and Arman Rouhani \Letter\orcidID{0000-0003-4822-484X} \and Marc Schr{\"o}der\orcidID{0000-0002-0048-2826}}
%

\institute{Department of Quantitative Economics, Maastricht University, The Netherlands 
\email{\{a.berger,a.rouhani,m.schroder\}@maastrichtuniversity.nl}}

\maketitle              

\begin{abstract}
In this paper, we introduce an improved upper bound for the efficiency of Nash equilibria in utilitarian scheduling games on related machines. The machines have varying speeds and adhere to the Shortest Processing Time (SPT) policy as the global order for job processing. The goal of each job is to minimize its completion time, while the social objective is to minimize the sum of completion times. Our main result provides an upper bound of $2-\frac{1}{2\cdot(2m-1)}$ on the price of anarchy for the general case of $m$ machines. We improve this bound to 3/2 for the case of two machines, and to $2-\frac{1}{2\cdot m}$ for the general case of $m$ machines when the machines have divisible speeds.

\end{abstract}
\section{Introduction}
Resource allocation is a fundamental and extensively studied problem in the domains of Operations Research and Economics. The aim of multi-agent resource allocation \cite{chevaleyre2005issues} is to achieve an efficient allocation of resources among self-interested agents. In numerous real-world scenarios, such as task assignments to processors, runway utilization by airplanes, vehicle parking, or allocation of classrooms for teaching purposes, the number of activities often exceeds the available resources. Consequently, a critical requirement for a scheduling mechanism that effectively allocates users to these limited resources arises. Inefficiency can arise due to the allocation procedure itself. Traditionally, a central decision-maker possesses complete information about the problem and computes an allocation that optimizes a specific objective. However, as the world witnesses a surge in distributed resources, there is an escalating demand for decentralized allocation procedures. In such systems, individual units tend to prioritize their own objectives over the overall system performance, making them more susceptible to strategic behavior. The concept of the \textit{price of anarchy} was introduced by Koutsoupias and Papadimitriou \cite{koutsoupias2009worst} to provide a standard framework for analyzing the inefficiency of equilibria. Since its inception, this concept has been widely adopted in various decentralized optimization problems, including network routing and scheduling. It serves as a valuable tool for quantifying the potential losses in system efficiency that arise when self-interested agents make independent decisions.

The price of anarchy is extensively used for measuring the inefficiency of scheduling games. The existing literature delves into diverse problem variants, encompassing a range of objectives such as minimizing the \textit{makespan} or \textit{total completion time} (egalitarian and utilitarian social cost functions \cite{myerson1981utilitarianism}). Additionally, it explores strategic behaviors exhibited by machines, as well as various job processing order policies employed by these machines. By examining these multiple dimensions, researchers gain valuable insights into the intricacies of scheduling systems. Extensive research has been conducted in the literature concerning the makespan objective, which focuses on minimizing the maximum load of machines. Notable works such as \cite{yu2010price,azar2015optimal,christodoulou2009coordination} have provided valuable insights in this area, also tight bounds for different versions of the problem are obtained \cite{awerbuch2004tradeoffs,czumaj2007tight,gairing2010computing}. Furthermore, there is a significant body of literature dedicated to minimizing the sum of completion times objective \cite{bhattacharya2014coordination,cole2015decentralized,correa2012efficiency,rahn2013bounding}, where the completion time of each job depends on the jobs processed earlier than it. As a result, the scheduling order on each machine becomes a crucial factor \cite{vijayalakshmi2021scheduling}. 

We are interested  in a version of the scheduling game where there are $n$ jobs that need to be processed on $m$ related machines with different speeds, where each job acts selfishly and aims to minimize its own completion time. The objective of the problem is to minimize the total sum of completion times of jobs where the machines obey a global priority policy for scheduling the jobs which is based on the \textit{Shortest Processing Time} first rule (a.k.a. Smith's rule). This is a continuation of the research line of Hoeksma and Uetz \cite{hoeksma2019price} where they proved that the price of anarchy lies in the interval of $[\frac{e}{e-1},2]$. The upper bound of 2 was improved to $2 - 2/(n^2 + n + nm + m)$ by Zhang et al. \cite{zhang2019improved} which asymptotically still leads to the same bound of 2 even for a fixed number of machines, and they provided a lower bound of 1.1875 for the case of two machines.  We improve the upper bound to $2-1/(2\cdot(2m-1))$ which only depends on the number of machines and to the best of our knowledge is the best current known upper bound for this problem. Moreover, we further improve this bound to $2-\frac{1}{2\cdot m}$ when the machines have divisible speeds. Lee et al. \cite{lee2012coordination} established that in the context of two machines, the price of anarchy is confined within the range of [$(3+\sqrt{3})/4$, $(1+\sqrt{5})/2$]. Strengthening this understanding, Brokkelkamp \cite{closePHD} subsequently enhanced the lower bound to 1.188. Notably, our advancements extend to refining the upper bound of the price of anarchy for this scenario to 3/2. It is shown in the literature \cite{heydenreich2007games} that the \textit{Ibarra-Kim} algorithm which is an approximation algorithm for the optimization variant of the problem, generates all pure strategy Nash equilibria. As a result, our presented bound can be viewed as an enhanced approximation ratio for this greedy algorithm. This improvement in the approximation ratio further solidifies the effectiveness of our proposed approach.

In the context of unrelated machines, previous investigations resulted in a tight bound of 4 for the price of anarchy for both weighted and unweighted settings \cite{cole2015decentralized,correa2012efficiency}. When the machines are identical, there are tight bounds of $1$ and $3/2-1/2m$ for the pure price of anarchy and mixed Nash equilibria, respectively \cite{rahn2013bounding,hoeksma2019price}. Furthermore, Brokkelkamp \cite{closePHD} investigated related machine scheduling games using the primal-dual technique and derived the same bounds as Hoeksma and Uetz \cite{hoeksma2019price}. Table \ref{table:results} summarizes the existing results for the pure price of anarchy bounds in utilitarian machine scheduling games with the sum of (weighted) completion times as the objective. In the scheduling notation $\alpha|\beta|\gamma$ introduced by Graham et al. \cite{graham1979optimization}, our problem is represented as $Q||\sum_{}^{}C_{j}$ and, in the context of a two-machine scenario, as $Q2||\sum_{}^{}C_{j}$. Additionally, $P||\sum_{}^{}C_{j}$ and $P||\sum_{}^{}w_{j}C_{j}$ delineate scenarios where machines are identical, with the latter incorporating weights $w_j$ associated with each job $j$. Finally, $R ||\sum w_j C_{j}$ signifies unrelated machine scheduling with weighted jobs.

\begin{table}
    \caption{Summary of existing results.}  \label{table:results}
  \centering
  \renewcommand{\arraystretch}{1.2}
  \begin{tabular}{|p{3cm}||c|c|c|}
    \hline
    \centering
    \multirow{2}{*}{\textbf{Problem Setting}} & \multicolumn{2}{c|}{\textbf{Pure POA}}\\
    \cline{2-3}
    & \textbf{LB} & \textbf{UB}\\
    \hline 
     $P||\sum_{}^{}C_{j}$ & 1 \cite{conway2003theory} & 1 \cite{conway2003theory}  \\ \hline 
     $P||\sum_{}^{}w_{j}C_{j}$ & $(1+\sqrt{2})/2 \approx 1.2$ \cite{kawaguchi1986worst} & $(1+\sqrt{2})/2\approx 1.2$ \cite{kawaguchi1986worst}\\ \hline
    $Q||\sum_{}^{}C_{j}$ & $e/(e-1) \approx 1.58$ \cite{hoeksma2019price} & 2 \cite{hoeksma2019price} \\     \hline 
   $Q2||\sum_{}^{}C_{j}$ & $1.188$ \cite{closePHD} & $\frac{1+\sqrt{5}}{2} \approx 1.618$ \cite{lee2012coordination} \\ \hline
     $R ||\sum w_j C_{j}$ & 4 \cite{correa2012efficiency} & 4\cite{cole2015decentralized} \\
     \hline
  \end{tabular} 
\end{table}
Section 2 begins by formally defining the problem and introducing the notations that will be utilized throughout the paper. Additionally, we provide concise summaries of the optimal schedule, the MFT algorithm, pure Nash equilibria, the price of anarchy, and the Ibarra-Kim algorithm. Furthermore, we present an alternative perspective of the problem used in \cite{closePHD} by introducing linear programming and its dual to analyze the problem. In Section 3, we present our approach, starting with the case of two machines and deriving a bound of 3/2. Subsequently, we extend our analysis to the general case and establish a bound of $2-1/(2\cdot(2m-1))$. Finally, in Section 4, we conclude our findings, highlight our contribution, and provide future directions.
\section{Preliminaries}
We are interested in finding an improved upper bound for the price of anarchy for scheduling games on related machines with Smith's rule as the global priority list for all machines. We study a variant of scheduling games that entails scheduling a collection of $n$ jobs from the set $J$ onto a set of $m$ machines from the set $M$. Denote $p_j$ as the processing time of job $j \in J$. If job $j$ is assigned to machine $i$, then the time needed to execute job $j$ on machine $i$ is equal to $p_j/s_i$, where  $s_i$ denotes the speed of machine $i$. To simplify the analysis, we normalize the speeds of the machines by dividing the speeds of all machines by the speed of the slowest machine. In this modified setting, we denote the speed of the fastest machine as $s_m$ or $s$, while keeping the speed of the slowest machine fixed at 1. Without loss of generality, we assume an ordering of the processing times as $p_1 \leq p_2 \leq \ldots \leq p_n$ and a similar ordering of the speeds as $s_1 = 1 \leq s_2 \leq \ldots \leq s_m = s$. This problem is denoted by $Q||\sum_{}^{}C_{j}$. Every job acts as a self-interested player and selects a machine for processing in order to minimize its own completion time $C_j$. Consequently, jobs consider the priority list of machines, which determines the sequence in which machines process the jobs that have chosen them. We investigate the price of anarchy resulting from a machine priority list, specifically the Shortest Processing Time or SPT order, commonly known as Smith's rule. This order arranges the jobs in non-decreasing order based on their processing time. We make the assumption that ties are resolved using a universally consistent total order across all machines. The strategy set for each job is equal to the set of machines $M$, indicating that the players can assign their jobs to every single machine. The strategies of the jobs are represented by a vector $\tau=(\tau_j)_{j\in J}$, where $\tau_j$ denotes the machine chosen by job $j$. The cost incurred by jobs in a strategy profile $\tau \in M^n$, where a job $j$ has chosen machine $\tau_j$, is determined by the completion times of the respective jobs $C_j(\tau)$. The completion time of job $j$ on machine $i$ is calculated as the sum of its processing time, $p_i$, and the processing times of all the jobs that select machine $i$ and that have a higher priority than job $j$. The social cost function based on strategy profile $\tau$ is defined as the aggregate sum of the completion times for all participants choosing their machines based on strategy $\tau$
\begin{equation}
     SC(\tau) = \sum_{j \in J}^{}C_j(\tau) = \sum_{j \in J}^{}\sum_{\substack{k\leq j \\ \tau_k = \tau_j}}\frac{p_k}{s_{\tau_j}}.
\end{equation}
Let $\tau$ represent a schedule, and consider a machine $i$ and a job $j$. Denote 
\begin{equation}\label{eq:zval}
    \psi_i(\tau,j) = |\{k>j | \tau_k = i\}|,
\end{equation}
as the symbol for the number of jobs on machine $i$ in schedule $\tau$ that have strictly lower priority than job $j$. It is worth noting that the sum of completion times can be reformulated as
\begin{equation}
     SC(\tau) = \sum_{j \in J}^{}p_j\cdot\frac{\psi_{\tau_j}(\tau,j) + 1}{s_{\tau_j}}.
\end{equation}
\subsection{Optimal Schedule for Jobs}
In this subsection, we provide a brief overview of the optimal algorithm for $Q||\sum_{}^{}C_{j}$.

Imagine the scenario where we have only one machine with speed 1 and we aim to minimize the sum of the completion times. From another perspective, the social cost based on a strategy profile $\tau$ can be measured by the cumulative contribution of each job to the total sum
\begin{equation}
      SC(\tau) = \sum_{k=1}^{n}(n-k+1)p_k.
\end{equation}
Therefore, the optimal schedule is a schedule that matches bigger jobs (higher processing times) to smaller values of $(n-k+1)$  in order to minimize the overall impact of the jobs on the social cost. The MFT algorithm is a generalization of this idea to the case of $m$ machines. The complexity of this algorithm is $O(n\log{nm})$ \cite{horowitz1976exact}.

\begin{algorithm}
\caption{MFT Algorithm for $Q||\sum_{}^{}C_{j}$}\label{alg:MFTQ}
\begin{algorithmic}
\State For each machine $i$ set $\psi_i = 0$.
\While{not all jobs are scheduled}
\State Let $j$ be the unscheduled job with largest index
\State Assign job $j$ to the machine with the smallest value of $(\psi_i + 1)/s_i$.
\State For that machine update $\psi_i = \psi_i + 1\cdot$
\EndWhile
\State Sort the jobs on each machine in SPT order.
\end{algorithmic}
\end{algorithm}
We assume that when encountering ties in choosing the smallest value of $(\psi_i + 1)/s_i$ during the execution of the MFT algorithm, the established tie-breaking criterion dictates selecting the machine with the highest index (and thus greatest speed).

\begin{theorem}[Horowitz and Sahni \cite{horowitz1976exact}]
The MFT (Mean Flow Time) algorithm produces optimal schedules for the related machine scheduling problem $Q||\sum_{}^{}C_{j}$ and any optimal schedule can be computed by the MFT algorithm with a proper tie-breaking rule.
\end{theorem}
Moreover, we introduce a lemma from Hoeksma and Uetz \cite{hoeksma2019price} that describes a useful characteristic of an optimal schedule and will be instrumental for our proofs in the later stages. In simple words, this lemma asserts that a job assigned to a machine cannot possess a better positional value than the jobs already allocated to any other machine.
\begin{lemma}[Hoeksma and Uetz \cite{hoeksma2019price}]\label{optINEQ}
A schedule $\tau^*$ is optimal for $Q||\sum_{}^{}C_{j}$ if and only if 
\begin{equation}\label{eq:IEQopt}
 \frac{\psi_i(\tau^*,j) + 1}{s_i} \geq \frac{\psi_l(\tau^*,j)}{s_l},
\end{equation}
for all machines $i$ and $l$, and for all jobs $j$.
\end{lemma}
\subsection{Pure Nash Equilibria and Price of Anarchy}
Define $\tau_{-j}$ as a vector obtained by removing $\tau_{j}$ from $\tau$, where $|\tau_{-j}| = n-1$. In other words, $\tau_{-j}$ represents the strategy vector of all agents except agent $j$. Therefore, we can express $\tau$ as the concatenation of $\tau_{-j}$ and $\tau_{j}$, denoted as $\tau = (\tau_{-j}, \tau_{j})$. For the problem $Q||\sum_{}^{}C_{j}$ a strategy profile $\tau = (\tau_{-j}, \tau_{j}) \in M^n$ is a \textit{(pure) Nash equilibrium} if 
\begin{equation}
C_j(\tau) = C_j(\tau_{-j}, \tau_{j}) =\sum_{\substack{k\leq j \\ \tau_k = \tau_j}}\frac{p_k}{s_{\tau_j}} \leq \frac{p_j}{s_{i}}+\sum_{\substack{k< j \\ \tau_k = i}}\frac{p_k}{s_{i}}= C_j(\tau_{-j}, i)  \tag{$\forall i \in M$, $\forall j \in J$}.
\end{equation}
It means that a Nash equilibrium is reached in a pure strategy profile when no job has the motivation to shift to a different machine because such a move would be futile in terms of reducing its completion time within the outcome schedule. 

\begin{algorithm}\label{alg:ibr-k}
\caption{Ibarra-Kim Algorithm for problem $Q||\sum_{}^{}C_{j}$}\label{alg:Ibarra-Kim}
Sort the jobs based on the non-decreasing order of their processing times.
\begin{algorithmic}
\While{not all jobs are scheduled}
\State Take job $k$ from the beginning of the sorted list and delete it from the list.
\State Let machine $i$ be the machine where job $k$ has minimal completion time.
\State Schedule job $k$ directly after the jobs already scheduled on machine $i$.
\EndWhile
\end{algorithmic}
\end{algorithm}

As we mentioned earlier, the global order for each machine is based on the SPT rule and the completion time of a job in a machine depends only on the jobs with higher priority scheduled on that machine. The well-known \textit{Ibarra-Kim} algorithm was first introduced by Ibarra and Kim \cite{ibarra1977heuristic} as an approximation algorithm for the optimization variant of the problem. This algorithm first sorts the jobs based on the non-decreasing order of their processing times (SPT rule) and then greedily assigns the jobs from the beginning of the sorted list to a machine that minimizes their completion time (see Algorithm \ref{alg:ibr-k}). The following theorem shows that the set of schedules produced by the Ibarra-Kim algorithm is exactly equal to the set of the pure strategy Nash equilibria, based on the manner in which ties are resolved.
\begin{theorem}[Heydenreich et al. \cite{heydenreich2007games}, Immorlica et al \cite{immorlica}]
Depending on how ties are resolved, the collection of pure Nash equilibria for unrelated (related, parallel) machines coincides precisely with the solutions derived from the Ibarra-Kim algorithm.
\end{theorem}
Hence, examining the price of anarchy in the related machine scheduling game is tantamount to analyzing the approximation factor of the Ibarra-Kim algorithm for related machine scheduling. 

Let $\chi$ be the set containing all possible pure Nash equilibria. Consider $\tau^* \in M^n$ as the strategy profile that minimizes the social cost. Then the \textit{(pure) price of anarchy} is defined as 
\begin{equation}
PoA = \sup_{\tau \in \chi}\frac{SC(\tau)}{SC(\tau^*)}\cdot
\end{equation}
The primal-dual method, introduced by Bilò \cite{bilo2018unifying}, is a strong approach for analyzing the price of anarchy and is able to effectively capture the characteristics of pure Nash equilibria. Consider the problem $Q||\sum_{}^{}C_{j}$ and assume a Nash equilibrium $\tau$ and an optimal strategy profile $\tau^*$ is given. Based on the primal-dual method, the aim is to construct a linear program $LP(\tau, \tau^*)$ that maximizes $SC(\tau)$ with respect to the processing times and subject to the constraints that $SC(\tau^*)$ is normalized to one and $\tau$ must be a Nash equilibrium. 
It is noteworthy that the processing times can be scaled down uniformly in a way that the social cost of the optimal schedule becomes equal to one. Based on weak duality, any feasible solution for the dual problem with an objective value provides an upper bound on the optimal primal solution and thus an upper bound on the price of anarchy. As mentioned in the related literature, Brokkelkamp \cite{closePHD} employed the dual-primal method to achieve the same bounds as Hoeksma and Uetz \cite{hoeksma2019price}. For the sake of completeness, we briefly introduce their weak linear programming formulation for the problem $Q||\sum_{}^{}C_{j}$. This version of LP is called ``weak" since it only allows the deviation of a job to the optimal machine for the Nash equilibrium constraints (and not to all other machines as required in a Nash equilibrium). This idea follows from the smoothness concept \cite{hoeksma2019price}. This idea follows from the smoothness concept \cite{hoeksma2019price}. The inputs of the LP are $s = (s_1,...,s_m)$, $\tau^*$, and $\tau$ where $s$ is the speed vector of machines and $\tau^*$, and $\tau$ are the optimal schedule and Nash equilibrium profile, respectively. Let $p = (p_1,...,p_n)$ denote the vector of processing times. The LP is obtained as follows

\begin{align*}
   &\phantom{} LP(\tau^*,\tau,s) = \max_{p } \sum_{j \in J}^{} \frac{1}{s_{\tau_j}}\sum_{\substack{k\leq j \\ \tau_k = \tau_j}}p_k \qquad s.t.\\
   &  \sum_{j \in J}^{} \frac{1}{s_{\tau^*_j}}\sum_{\substack{k\leq j \\ \tau^*_k = \tau^*_j}}p_k = 1, \tag{normalization constraint} \\
   & \frac{1}{s_{\tau_j}} \left[  p_j + \sum_{\substack{k < j \\ \tau_k = \tau_j}}p_k \right] -  \frac{1}{s_{\tau^*_j}} \left[  p_j + \sum_{\substack{k < j \\ \tau_k = \tau^*_j}}p_k \right] \leq 0, \tag{deviation constraints ; $\forall j \in J$}\\
   & p_j - p_{j+1} \leq 0, \tag{SPT constraints ;  $\forall j \in \{1,..,n-1\}$}\\
   & -p_1 \leq 0.\\
\end{align*}

In order to obtain the dual of this linear program, the dual variables are defined as $\beta$ for the normalization constraint, $y_j$ variables for each of the deviation constraints, and $z_1$ to $z_{n-1}$ for the SPT constraints. Note that $z_n = z_0 = 0$. The dual is obtained as follows

\begin{align*}
&\hspace*{0pt} \text{Dual}(\tau^*,\tau,s) = \min \; \beta \quad \text{s.t.}\\
&  \beta \cdot \frac{1}{s_{\tau^*_k}} \sum_{\substack{j \geq k \\ \tau^*_k = \tau^*_j}}1 +
   \frac{1}{s_{\tau_k}} \sum_{\substack{j \geq k \\ \tau_j = \tau_k}}y_j -
   \frac{y_k}{s_{\tau^*_k}} -
   \frac{1}{s_{\tau_k}} \sum_{\substack{j > k \\ \tau^*_j = \tau_k}}y_j +
   z_k - z_{k-1} \geq \frac{1}{s_{\tau_k}} \sum_{\substack{j \geq k \\ \tau_j = \tau_k}}1, && \tag{$\forall k \in J$}\\
& y_j \geq 0, && \tag{$\forall k \in J$}\\
& z_k \geq 0. && \tag{$\forall k \in \{1,..,n-1\}$}
\end{align*}

Henceforth in this study, we will refer to $LP(\tau, \tau^*)$ and $Dual(\tau, \tau^*)$ as $LP_{weak}$ and $Dual_{weak}$, respectively. An upper bound of 2 can be easily obtained by a feasible solution of $Dual_{weak}$ given by $\beta = 2$, $y_j = 1$, and $z_j = 0$ for all $j \in J$ \cite{closePHD}. However, for the worst-case instance described in \cite{hoeksma2019price}, the optimal dual fitting for the weak LP yields a value of 5/3. It is known that this value can be as low as e/(e - 1) \cite{hoeksma2019price}. In this study, we use a more relaxed version of the $LP_{weak}$ to improve the existing upper bound of 2.

\section{An upper bound for the price of anarchy of $Q||\sum_{}^{}C_{j}$}

In this section, we improve the existing upper bound for the problem $Q||\sum_{}^{}C_{j}$. First, we analyze the case where the number of machines is equal to 2 and prove an upper bound of $3/2$ for the pure price of anarchy. Then, we provide an example that shows that our analysis is tight for this case. Next, we generalize the upper bound to $2-\frac{1}{2\cdot(2m-1)}$ for the case of $m$ machines and we also show that when the machines have divisible speeds this bound can be further improved to $2-\frac{1}{2\cdot m}$. 

Given $\tau^*$ and $\tau$, an optimal schedule and a Nash equilibrium schedule, respectively, the optimal objective value of the following linear program is lower bounded by the sum of completion times of Nash equilibrium $\tau$. Therefore, we can use its dual to obtain an upper bound on the price of anarchy. Define $LP_{sum}$ as follows
\begin{align*}
   &\phantom{} LP(\tau^*,\tau,s) = \max_{p} \sum_{j \in J}^{} \frac{1}{s_{\tau_j}}\sum_{\substack{k\leq j \\ \tau_k = \tau_j}}p_k \qquad s.t.\\
   &  \sum_{j \in J}^{} \frac{1}{s_{\tau^*_j}}\sum_{\substack{k\leq j \\ \tau^*_k = \tau^*_j}}p_k = 1, \tag{normalization constraint} \\
   & \sum_{j \in J}^{} \left[ \frac{1}{s_{\tau_j}} \left( p_j + \sum_{\substack{k < j \\ \tau_k = \tau_j}}p_k \right) \right] -  \sum_{j \in J}^{} \left[ \frac{1}{s_{\tau^*_j}} \left(  p_j + \sum_{\substack{k < j \\ \tau_k = \tau^*_j}}p_k  \right) \right] \leq 0, \tag{deviation constraint}\\
   & p_j - p_{j+1} \leq 0, \tag{ SPT constraints ; $\forall j \in \{1,..,n-1\}$}\\
   & -p_1 \leq 0.
\end{align*}

To obtain the dual linear program $Dual_{sum}$ of $LP_{sum}$, the dual variables are defined as $\beta$ for the normalization constraint, $y$ for the deviation constraint, and $z_1$ to $z_{n-1}$ for the SPT constraints. Note that $z_n = 0$. $Dual_{sum}$ is defined as follows

\begin{align*}
&\hspace*{0pt} \text{Dual}(\tau^*,\tau,s) = \min \; \beta \quad \text{s.t.}\\
&  \beta \cdot \frac{1}{s_{\tau^*_k}} \sum_{\substack{j \geq k \\ \tau^*_k = \tau^*_j}}1 +
   \frac{y}{s_{\tau_k}} \sum_{\substack{j \geq k \\ \tau_j = \tau_k}}1 -
   \frac{y}{s_{\tau^*_k}} -
   \frac{y}{s_{\tau_k}} \sum_{\substack{j > k \\ \tau^*_j = \tau_k}}1 +
   z_k - z_{k-1} \geq \frac{1}{s_{\tau_k}} \sum_{\substack{j \geq k \\ \tau_j = \tau_k}}1, && \tag{$\forall k \in J$}\\
& z_k \geq 0. && \tag{$\forall k \in \{1,..,n-1\}$}\\
& y \geq 0.
\end{align*}

Upon careful examination of $Dual_{sum}$, it becomes evident that $Dual_{sum}$ can be derived from $Dual_{weak}$ by uniformly setting all $y_j$ values to $y$. This is in line with the dual solution proposed by Brokkelkamp \cite{closePHD}: $\beta = 2$, $y_j = 1$, and $z_j = 0$ for all $j \in J$. In order to improve this dual solution, we maintain the assumption that $y=1$, so that the second term on the left-hand side of the deviation constraint of $Dual_{sum}$ cancels with the term on the right-hand side of the inequality. Additionally, we will define some $z_j$ variables to take non-zero values in order to attain a bound that is strictly below 2 for a fixed number of machines. 

We first assume the number of machines is equal to 2 and provide a bound of 3/2. 

\begin{theorem}\label{thm3}
The pure price of anarchy for $Q2||\sum_{}^{}C_{j}$ is bounded above by 3/2.
\end{theorem}

\begin{proof}
We will show that $\beta = 3/2$, $y = 1$, $z_{n-\lfloor s \rfloor} =\frac{\lfloor s \rfloor}{s}\cdot\frac{\lfloor s \rfloor}{s+\lfloor s \rfloor}$, and $z_j = 0$ for all $j\neq n-\lfloor s \rfloor$,  constitutes a feasible solution for $Dual_{sum}$, valid for all possible values of $\tau^*$, $\tau$, and $s$ and thus yields an upper bound for the price of anarchy. By utilizing the definition introduced by equation (\ref{eq:zval}), we rewrite the dual constraint in the following manner
\begin{alignat*}{2}
\beta \cdot \frac{\psi_{\tau^*_k}(\tau^*,k)+1}{s_{\tau^*_k}} &+ y \cdot \frac{\psi_{\tau_k}(\tau,k)+1}{s_{\tau_k}} - \frac{y}{s_{\tau^*_k}} \\
&- y \cdot \frac{\psi_{\tau_k}(\tau^*,k)}{s_{\tau_k}} +z_k - z_{k-1}&\geq \frac{\psi_{\tau_k}(\tau,k)+1}{s_{\tau_k}}.\tag{$\forall k \in J$} 
\end{alignat*}
By setting $y=1$ and streamlining the dual constraint, we obtain a lower bound for $\beta$ as follows
\begin{equation}\label{eq:GIN}
\beta \geq \frac{1}{\psi_{\tau^*_k}(\tau^*,k)+1} + \frac{\psi_{\tau_k}(\tau^*,k)}{s_{\tau_k}} \cdot \frac{s_{\tau^*_k}}{\psi_{\tau^*_k}(\tau^*,k)+1} +(z_{k-1}- z_k)\cdot \frac{s_{\tau^*_k}}{\psi_{\tau^*_k}(\tau^*,k)+1}.
\end{equation}
Moving forward, we direct our attention towards demonstrating that inequality (\ref{eq:GIN}) is satisfied for all jobs.

The optimal schedule, determined by the MFT algorithm, can be visualized as a block structure shown in Figure \ref{figblock}. Each block in the fast machine consists of a set of jobs that are consecutively assigned to the same machine. Based on the  MFT algorithm, we can observe that Block 1 contains $\lfloor s \rfloor$ longest jobs. The remaining blocks on the fast machine contain either $\lfloor s \rfloor$ or $\lceil s \rceil$ jobs. Furthermore, each block on the slow machine contains precisely one job. Initially, we examine the feasibility of our dual solution for jobs $k \in J$, where $k < n-\lfloor s \rfloor$ (refer to the shaded area of Figure \ref{figblock} for an illustration). In accordance with our dual solution, $z_k=0$ for all $k < n-\lfloor s \rfloor$. Thus, we have
\begin{figure}[t!]
\centering
\includegraphics[scale=0.7]{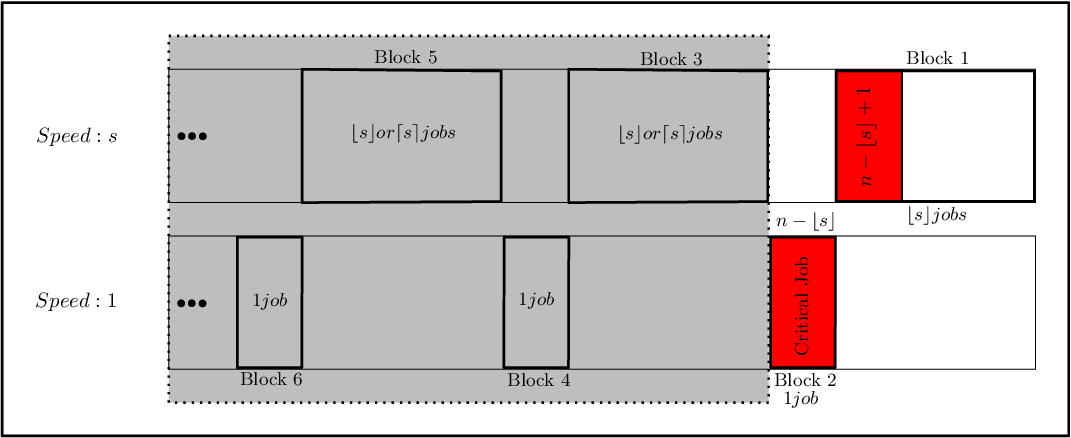}
\caption{The block structure of any optimal schedule when $m=2$.} \label{figblock}
\end{figure}

\begin{alignat*}{2}
\beta= \frac{3}{2}\geq \frac{1}{\psi_{\tau^*_k}(\tau^*,k) + 1} + 1\geq \frac{1}{\psi_{\tau^*_k}(\tau^*,k) + 1} + \frac{\psi_{\tau_k}(\tau^*,k)}{\psi_{\tau^*_k}(\tau^*,k) + 1}\cdot \frac{s_{\tau^*_k}}{s_{\tau_k}},
\end{alignat*}
where the first inequality follows from the fact that $\psi_{\tau^*_k}(\tau^*,k)\geq 1$ for these jobs (see Figure \ref{figblock}), and the second inequality from inequality (\ref{eq:IEQopt}). Hence, inequality (\ref{eq:GIN}) is satisfied for all $k < n - \lfloor s \rfloor$. Now, it remains to extend our analysis for jobs $k \in J$, where $k \geq n-\lfloor s \rfloor$. Let us refer to the job $n-\lfloor s \rfloor$ as the \textit{\textbf{critical job}} (see Figure \ref{figblock}). For the critical job we know that $\psi_{\tau^*_{n-\lfloor s \rfloor}}(\tau^*,{n-\lfloor s \rfloor}) = 0$, $s_{\tau^*_{n-\lfloor s \rfloor}}= 1$, $z_{{n-\lfloor s \rfloor}-1}=0$, and $z_{n-\lfloor s \rfloor} =\frac{\lfloor s \rfloor}{s}\cdot\frac{\lfloor s \rfloor}{s+\lfloor s \rfloor}$. We know that $1 \leq \frac{s}{\lfloor s \rfloor}$ for $s \geq 1$. Thus
\begin{equation}\label{sUP}
  1  \leq \frac{s}{\lfloor s \rfloor} \Leftrightarrow 1 + 1 \leq \frac{s + \lfloor s \rfloor}{\lfloor s \rfloor} \Leftrightarrow \frac{1}{2} \geq \frac{\lfloor s \rfloor}{s+\lfloor s \rfloor}.
\end{equation} 
Moreover, in the Nash equilibrium, the value of $\psi_{\tau_{n-\lfloor s \rfloor}}(\tau^*,n-\lfloor s \rfloor)$ is zero if the critical job is assigned to the slow machine, while $\psi_{\tau_{n-\lfloor s \rfloor}}(\tau^*,n-\lfloor s \rfloor) = \lfloor s \rfloor$ if the critical job is assigned to the fast machine. By substituting the relevant values into inequality (\ref{eq:GIN}) and leveraging the previous inequality, we derive the following bound for the critical job.
\begin{alignat*}{2}
\quad  \beta = \frac{3}{2} = 1 + \frac{1}{2}\geq 1+\frac{\lfloor s \rfloor}{s+\lfloor s \rfloor} &= 1+\frac{\lfloor s \rfloor}{s}-\frac{\lfloor s \rfloor}{s}\cdot\frac{\lfloor s \rfloor}{s+\lfloor s \rfloor}\\
&\geq 1+\frac{\psi_{\tau_{n-\lfloor s \rfloor}}(\tau^*,n-\lfloor s \rfloor)}{s_{\tau_{n-\lfloor s \rfloor}}}- z_{n-\lfloor s \rfloor}.
\end{alignat*}  
Thus, inequality (\ref{eq:GIN}) is satisfied. Now, let us focus on analyzing the jobs in Block 1 ($k > n-\lfloor s \rfloor$). These jobs are always assigned to the fast machine in the optimal schedule. By considering potential machines that can used in the Nash equilibria, we can simplify inequality (\ref{eq:GIN}) and see which one results in a stronger bound for $\beta$. If $\tau_k$ = $1$, then, $\psi_{\tau_k}(\tau^*,k) = 0$ and inequality (\ref{eq:GIN}) becomes
\begin{alignat*}{2}
\quad & \beta \geq \frac{1}{\psi_{\tau^*_k}(\tau^*,k)+1} +(z_{k-1}- z_k)\cdot \frac{s_{\tau^*_k}}{\psi_{\tau^*_k}(\tau^*,k)+1}.
\end{alignat*}
If $\tau_k$ = $2$, inequality (\ref{eq:GIN}) becomes
\begin{alignat*}{2}
\beta & \geq \frac{1}{\psi_{\tau^*_k}(\tau^*,k)+1} + \frac{\psi_{\tau^*_k}(\tau^*,k)}{\psi_{\tau^*_k}(\tau^*,k)+1} \cdot \frac{s_{\tau^*_k}}{s_{\tau^*_k}} +(z_{k-1}- z_k)\cdot \frac{s_{\tau^*_k}}{\psi_{\tau^*_k}(\tau^*,k)+1} \\
&= 1 +(z_{k-1}- z_k)\cdot \frac{s_{\tau^*_k}}{\psi_{\tau^*_k}(\tau^*,k)+1}.
\end{alignat*}
Hence, it becomes evident that when the machine that is used in the Nash equilibrium coincides with the fast machine, we obtain a stronger bound. Therefore, for the jobs in this block, our focus is on the latter inequality. Based on inequality (\ref{eq:GIN}), $z_{n-\lfloor s \rfloor}$ has an impact on the job $j_{n-\lfloor s \rfloor+1}$ which is in Block 1. Based on what we discussed earlier and according to our dual solution, for this particular job, we have that inequality (\ref{eq:GIN}) is satisfied because 
\begin{alignat*}{2}
\quad  \beta=\frac{3}{2} = 1 + \frac{1}{2} \geq 1+\frac{\lfloor s \rfloor}{s+\lfloor s \rfloor} &= 1 + z_{n-\lfloor s \rfloor}\cdot \frac{s}{\lfloor s \rfloor} \\
&= 1 +z_{n-\lfloor s \rfloor}\cdot \frac{s_{\tau^*_{n-\lfloor s \rfloor+1}}}{\psi_{\tau^*_{n-\lfloor s \rfloor+1}}(\tau^*,{n-\lfloor s \rfloor+1})+1},
\end{alignat*}
where the first inequality follows from $\frac{\lfloor s \rfloor}{s}\leq 1$ and $1+\frac{\lfloor s \rfloor}{s+\lfloor s \rfloor}$ being increasing in $\frac{\lfloor s \rfloor}{s}$ (see inequality (\ref{sUP})).
Similarly, according to our dual solution, the dual SPT constraints associated with all other jobs in Block 1 are set to zero. Consequently, for these remaining jobs, we obtain the inequality of $\beta \geq 1$, which is satisfied. Finally, we conclude that the solution $\beta = 3/2$, $y = 1$, $z_{n-\lfloor s \rfloor} =\frac{\lfloor s \rfloor}{s}\cdot\frac{\lfloor s \rfloor}{s+\lfloor s \rfloor}$, and $z_j = 0$ for all $j \in J$ is a feasible solution for  $Dual_{sum}$ that is valid for all possible values of $\tau^*$, $\tau$, and $s$.
\qed
\end{proof}
Consider that when $n \leq \lfloor s \rfloor$, the MFT algorithm schedules all jobs on the fastest machine, resulting in the dual variable associated with the SPT constraints being set to zero for all jobs. Henceforth, without loss of generality, we assume $n > \lfloor s \rfloor$.

The previous theorem establishes an upper bound of 3/2 for two machines. To demonstrate that we cannot improve this upper bound using $Dual_{sum}$, we present an example that illustrates the tightness of the bound for this method.

\begin{example}\label{ex1}
Consider the instance depicted in Figure \ref{eximg}. There are two jobs $j_1$ and $j_2$ with the processing times of $p_1$ and $p_2$, where $p_1\leq p_2$. Moreover, there are two identical machines with a speed of 1. Let $\tau = (2,2)$ be the vector of the Nash equilibrium and $\tau^*=(1,2)$ be the vector of the optimal schedule (see Figure \ref{eximg}). $LP_{sum}$ and $Dual_{sum}$ are presented in the following.
\begin{figure}[H]
    \centering
    \includegraphics[scale=0.5]{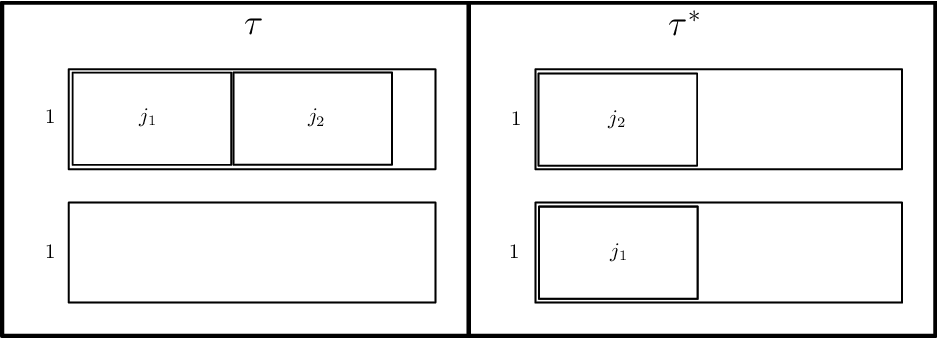}
    \caption{Schedules of Example \ref{ex1}.} \label{eximg}
    \begin{minipage}[t]{0.5\textwidth}
        \begin{align*}
            LP(\tau^*,\tau,s) = &\max \quad 2\cdot p_1 + p_2\\
            s.t. \quad &p_1 + p_2 = 1, \\
            &p_1\leq p_1,\\
            &p_1 - p_2 \leq 0, \\
            &-p_1 \leq 0.
        \end{align*}
    \end{minipage}%
    \begin{minipage}[t]{0.5\textwidth}
        \begin{align*}
            Dual(\tau^*,\tau,s) = &\min \quad \beta \\
             s.t. \quad &\beta + z_1 \geq 2, \\
            &\beta - z_1 \geq 1,\\
            &z_1 \geq 0.
        \end{align*}
    \end{minipage}
\end{figure}
By adding the two dual constraints, it is clear that $\beta \geq 3/2$. By setting $\beta =3/2$ and $z_1 =1/2$, we have a feasible (and thus optimal) solution that attains 3/2. Note that the weak LP does not impose a constraint for job $2$ and that is why it is not able to prove the actual price of anarchy of the example which is equal to 1.
\end{example}

In the previous analysis, for two machines, we established an upper bound of 3/2. Now, we extend the analysis for the general case of $m$ machines and we obtain a slightly weaker bound of $2-\frac{1}{2\cdot(2m-1)}$.

\begin{lemma}\label{floor}
Let $a, b \in \mathbb{R}$ with $a\geq b \geq 1$. Then, we have $\frac{a}{2b} \leq \lfloor \frac{a}{b} \rfloor$.
\end{lemma}
\begin{proof}
We can divide the analysis into two cases.

\textit{(1) $ 2b \leq a$}.

\begin{alignat*}{2}
\lfloor \frac{a}{b} \rfloor - \frac{a}{2b} = \frac{a}{b} - x - \frac{a}{2b} = \frac{a}{2b} - x \geq 0,
\end{alignat*}
where $0\leq x < 1$ is the decimal part of $\frac{a}{b}$. Additionally, building upon the given assumption, it is evident that $\frac{a}{2b} \geq 1$. Therefore, the last inequality holds.

\textit{(2) $ b \leq a < 2b$}.
\begin{alignat*}{2}
1 \leq \frac{a}{b} < 2 \Rightarrow \frac{1}{2} \leq \frac{a}{2b} < 1 \Rightarrow \frac{a}{2b} < \lfloor \frac{a}{b} \rfloor,
\end{alignat*}
where the last inequality holds since $\lfloor \frac{a}{b} \rfloor \geq 1$.
\qed
\end{proof}

\begin{lemma}\label{mftChar}
Let $k-1, k \in J$ be two consecutive jobs with $\tau^*_k\neq m$. Then, we have $\tau^*_{k-1} \neq \tau^*_k$.
\end{lemma}

\begin{proof}
Recall that the MFT algorithm first assigns job $k$ and then assigns job $k-1$. Assume that the MFT algorithm assigns job $k$ to machine $i$. As $k$ stands as the most recently scheduled job, its positional value cannot be smaller than previously scheduled jobs. Hence,
\begin{alignat*}{2}
\frac{\psi_{l}(\tau^*,k)}{s_l} \leq \frac{\psi_{i}(\tau^*,k)+1}{s_i}.\tag{$\forall l \in M$} 
\end{alignat*}
Thus, using the previous inequality, we observe that for all $l \geq i$,
\begin{alignat*}{2}
\frac{\psi_{l}(\tau^*,k)+1}{s_l} \leq \frac{\psi_{l}(\tau^*,k)}{s_l} + \frac{1}{s_i} \leq \frac{\psi_{i}(\tau^*,k)+2}{s_i}.
\end{alignat*}
The term on the left-hand side signifies the positional value of job $k-1$ on machine $\ell$ and the term on the right-hand side signifies the positional value of job $k-1$ on machine $i$. Consequently, we can infer that in the MFT algorithm, it is unfeasible to position two successive jobs adjacent to each other on the same machine, except for the fastest one. It is important to emphasize that this proof exclusively establishes the impossibility of placing two consecutive jobs on a single machine, as a more favorable positional value can consistently be found for the subsequent job among the higher-speed machines. Nonetheless, it's worth noting that the subsequent job could potentially be accommodated on a machine with a lower speed, contingent upon the positional value.
\qed
\end{proof}

\begin{theorem}\label{thm4}
The pure price of anarchy for $Q||\sum_{}^{}C_{j}$ is bounded above by $2-\frac{1}{2\cdot(2m-1)}$, where $m$ is the number of machines.
\end{theorem}

\begin{figure}[!t]
\centering
\includegraphics[scale= 0.6]{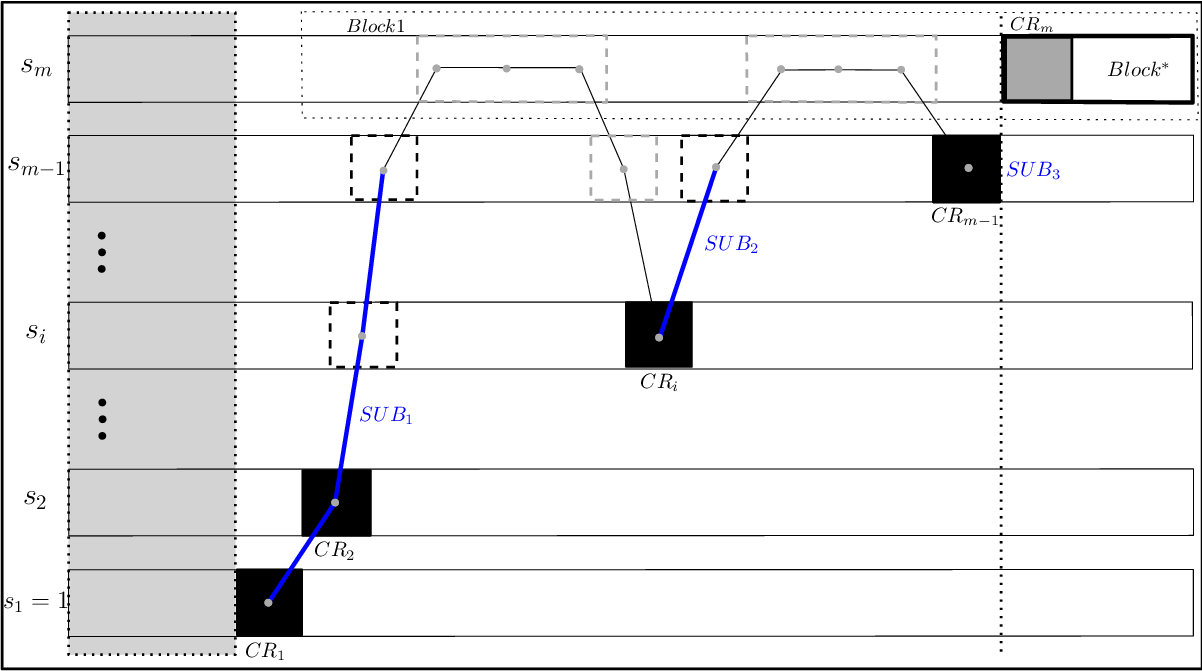}
\caption{A potential chain of consecutive jobs from $CR_1$ to $CR_m$ in an optimal schedule together with its sub-chains. Only jobs that are part of a sub-chain (blue lines) have a positive $z_j$-value.} \label{figblock3}
\end{figure}

\begin{proof}
Define the set of \textit{\textbf{critical jobs}} $CR$, one for each machine, where each $CR_i \in CR$ for $i=1,\ldots,m-1$ is the last job scheduled according to the optimal schedule on every machine $i$ and $CR_m\in CR$ is the smallest job in $Block^*$ (see Figure \ref{figblock3} for an illustration), where $Block^*$ contains all the consecutive largest jobs before the first job is assigned to the second fastest machine by the MFT algorithm. The significance of the last job scheduled on each machine becomes crucial in our analysis. 

We recursively define sub-chains as follows. Sub-chain 1, denoted by $SUB_1$, is the set of jobs $\{CR_1,CR_1+1,\ldots,j-1\}$, where $j=\min\{j'\geq CR_1\mid \tau^*_{j'}=m\}$. Notice that a sub-chain might contain multiple critical jobs. Then, sub-chain $i$, denoted by $SUB_i$, for $i\geq 2$ is the set of jobs $\{CR_{i'},CR_{i'}+1,\ldots,j-1\}$, where $CR_{i'}$ is the smallest indexed critical job that is not part of a previous sub-chain and $j=\min\{j'\geq CR_{i'}\mid \tau^*_{j'}=m\}$.
Let $SUB$ denote the set of all sub-chains and $\mathcal{SUB} = \bigcup SUB_i$ for $i\geq 1$. For $j \in \mathcal{SUB}$ denote by $r_j$ the machine index of the initial critical job within the sub-chain that includes job $j$.  We will show that
\begin{alignat*}{2}
    \beta &= 2 - \frac{1}{2 \cdot (2m - 1)}, \\
    y &= 1, \\
    z_j &=
    \begin{cases}
        i \cdot \frac{1}{2\cdot s_{r_j} \cdot(2m - 1)} & \text{, if }j\in\mathcal{SUB} \text{ and } CR_i \leq j < CR_{i+1} \text{ for some } i \in \{1, \ldots, m-1\}\\
        0 & \text{, otherwise}
    \end{cases}
\end{alignat*}
constitutes a feasible solution for $Dual_{sum}$, valid for all possible values of $s$, $\tau^*$, and $\tau$, and thus yields an upper bound for the price of anarchy. By this definition, only jobs $j$ that are part of a sub-chain will have a positive $z_j$ value, and the value depends on the sub-chain and on the position of that job between critical jobs within the sub-chain.
Similar to the proof of Theorem~\ref{thm3} we will now show that (\ref{eq:GIN}) holds for all jobs $k \in J$.

For jobs $k \in J$ with $k < CR_1$ (all the jobs in the shaded area of Figure \ref{figblock3}), we have that
\begin{alignat*}{2}
\beta\geq \frac{3}{2}\geq \frac{1}{\psi_{\tau^*_k}(\tau^*,k) + 1} + 1\geq \frac{1}{\psi_{\tau^*_k}(\tau^*,k) + 1} + \frac{\psi_{\tau_k}(\tau^*,k)}{\psi_{\tau^*_k}(\tau^*,k) + 1}\cdot \frac{s_{\tau^*_k}}{s_{\tau_k}},
\end{alignat*}
where the second inequality follows from $\psi_{\tau^*_k}(\tau^*,k)\geq 1$, and the third inequality from inequality (\ref{eq:IEQopt}). Since $z_{k-1}-z_k=0$, inequality (\ref{eq:GIN}) is satisfied for all $k < n - \lfloor s \rfloor$.

For all $k$ with $CR_1 < k < CR_m$, $\tau^*_k \neq m$, and $k\neq CR_i$ for $i=2,\ldots,m-1$, we have that
\begin{alignat*}{2}
\beta\geq \frac{3}{2}\geq \frac{1}{\psi_{\tau^*_k}(\tau^*,k) + 1} + 1\geq \frac{1}{\psi_{\tau^*_k}(\tau^*,k) + 1} + \frac{\psi_{\tau_k}(\tau^*,k)}{\psi_{\tau^*_k}(\tau^*,k) + 1}\cdot \frac{s_{\tau^*_k}}{s_{\tau_k}},
\end{alignat*}
where the second inequality follows from $\psi_{\tau^*_k}(\tau^*,k)\geq 1$, and the third inequality from inequality (\ref{eq:IEQopt}). Since $z_{k-1}-z_k=0$ for all these jobs, inequality (\ref{eq:GIN}) is satisfied. It is important to note that the dual variables of the SPT constraints for these jobs are non-zero if they are in a sub-chain and zero otherwise. 

Our focus now shifts to verifying the feasibility of our dual fitting concerning the dual constraints of both the critical jobs and the jobs scheduled on the fastest machine. Based on the definition of a sub-chain, each critical job can either serve as the starting point of a sub-chain or reside within a sub-chain originating from a prior critical job.

\begin{lemma}\label{minJobSUB}
Consider job $j$ with $\tau^*_j=m$ where job $j-1$ is the end-point of the sub-chain $SUB_i$. Then $\psi_{m}(\tau^*,j)+1  = \lfloor \frac{s_m}{s_i} \rfloor$.
\end{lemma}

\begin{proof}
The sub-chain $SUB_i$ originates from the critical job $CR_i$ on machine $i$. Since $j - 1$ is the end-point of the sub-chain, we have $ \psi_{m}(\tau^*,j)+1 =\psi_{m}(\tau^*,CR_i)$. Based on the MFT algorithm we have for job $j$ and $CR_i$, respectively,
\begin{align*}
        &\frac{\psi_{m}(\tau^*,CR_i)}{s_m} \leq \frac{1}{s_i}, \\
        &\frac{1}{s_i} \leq \frac{\psi_{m}(\tau^*,CR_i) + 1}{s_m}.
\end{align*}
By using both inequalities we have

\begin{alignat*}{2}
    \frac{s_m}{s_i} - 1 \leq \psi_{m}(\tau^*,CR_i) \leq \frac{s_m}{s_i}.
\end{alignat*}
Using $\frac{s_m}{s_i} - 1 \leq \lfloor \frac{s_m}{s_i} \rfloor \leq \frac{s_m}{s_i}$, and the fact that both $\psi_{m}(\tau^*,CR_i)$ and $\lfloor \frac{s_m}{s_i} \rfloor$ are integers, we can conclude
\begin{alignat*}{2}
\psi_{m}(\tau^*,CR_i) = \lfloor \frac{s_m}{s_i} \rfloor.
\end{alignat*}\qed
\end{proof}

For $k \in CR$ with $k \neq CR_m$, we have that
\begin{alignat*}{2}
\beta = 2-\frac{1}{2\cdot(2m-1)} &\geq 
2 + (z_{k-1}- z_k)\cdot s_{\tau^*_k} \\
&\geq \frac{1}{\psi_{\tau^*_k}(\tau^*,k) + 1} + \frac{\psi_{\tau_k}(\tau^*,k)}{\psi_{\tau^*_k}(\tau^*,k) + 1}\cdot \frac{s_{\tau^*_k}}{s_{\tau_k}}+(z_{k-1}- z_k)\cdot s_{\tau^*_k},
\end{alignat*}
using $z_{k-1}- z_k=-\frac{1}{2\cdot s_{r_k}\cdot(2m-1)}$, $s_{\tau^*_k}\geq s_{r_k}$, $\psi_{\tau^*_k}(\tau^*,k)=0$, and inequality (\ref{eq:IEQopt}).
Hence, inequality (\ref{eq:GIN}) is satisfied for all $k \in CR$ with $k \neq CR_m$. Now, let us focus on analyzing the jobs in the fastest machine in the MFT algorithm. Let $F$ be the set of all jobs $k$ that are allocated to the fastest machine and where job $k-1$ belongs to a sub-chain. For $k \in F$ where $k$ is not in $Block^*$, using the fact that $\psi_{\tau^*_k}(\tau^*,k)\geq 1$, and inequality (\ref{eq:IEQopt}) we have
\begin{alignat*}{2}
\beta &\geq \frac{3}{2} +(z_{k-1}- z_{k})\cdot \frac{s_{\tau^*_{k}}}{\psi_{\tau^*_{k}}(\tau^*,{k})+1}\\
&\geq \frac{1}{\psi_{\tau^*_k}(\tau^*,k) + 1} + \frac{\psi_{\tau_k}(\tau^*,k)}{\psi_{\tau^*_k}(\tau^*,k) + 1}\cdot \frac{s_{\tau^*_k}}{s_{\tau_k}}+(z_{k-1}- z_{k})\cdot \frac{s_{\tau^*_{k}}}{\psi_{\tau^*_{k}}(\tau^*,{k})+1}.
\end{alignat*}
Moreover, for $k \in F$ where $k$ is in the $Block^*$, using the same analysis as the case of two machines we have
\begin{alignat*}{2}
\beta & \geq 1 +(z_{k-1}- z_k)\cdot \frac{s_{\tau^*_k}}{\psi_{\tau^*_k}(\tau^*,k)+1}\\ &= \frac{1}{\psi_{\tau^*_k}(\tau^*,k)+1} + \frac{\psi_{\tau^*_k}(\tau^*,k)}{\psi_{\tau^*_k}(\tau^*,k)+1} \cdot \frac{s_{\tau^*_k}}{s_{\tau^*_k}} +(z_{k-1}- z_k)\cdot \frac{s_{\tau^*_k}}{\psi_{\tau^*_k}(\tau^*,k)+1}.
\end{alignat*}
Hence, the former inequality provides a more rigorous constraint and we consider that for the jobs in the fastest machine. Thus, we have
\begin{alignat*}{2}
\beta= 2-\frac{1}{2\cdot(2m-1)} &= \frac{3}{2} + \frac{m-1}{2\cdot s_{r_k}\cdot(2m-1)}\cdot (2\cdot s_{r_k})\\
& \geq \frac{3}{2} + z_{k-1} \cdot 2\cdot s_{r_k} \geq \frac{3}{2} +(z_{k-1}- z_{k})\cdot \frac{s_{\tau^*_{k}}}{\psi_{\tau^*_{k}}(\tau^*,{k})+1},
\end{alignat*}
where the first inequality follows from the fact that maximum $(m-1)$ critical jobs can cause increasing the value of the dual variable related to the SPT constraints and $r_k = r_{k-1}$ since they both are in the same sub-chain, and the second inequality is based on Lemma \ref{floor}, Lemma \ref{minJobSUB}, and the fact that $\psi_{\tau^*_{k}}(\tau^*,{k})+1 = \lfloor \frac{s_{\tau^*_{k}}}{s_{r_k}} \rfloor$. 

Hence, inequality (\ref{eq:GIN}) is satisfied for $k \in F$. For the remaining jobs in the fastest machine, since the dual variables related to the SPT constraints are zero, we obtain the inequality of $\beta \geq \frac{3}{2}$, which is satisfied. We conclude that the introduced dual fitting is a feasible solution for the $Dual_{sum}$ that is valid for all possible values of $\tau^*$, $\tau$, and $s$. 
\qed
\end{proof}

If for every pair of machines $i$ and $i^{'}$ belonging to $M$, where $s_i \geq s_{i^{'}}$, the remainder of the division $\frac{s_i}{s_{i^{'}}}$ is zero, then the configuration is referred to as having divisible speeds.

\begin{theorem}\label{thm:div}
In the case of machines having divisible speeds, the upper bound for the pure price of anarchy in the scheduling game $Q||\sum_{}^{}C_{j}$, is equal to $2-\frac{1}{2\cdot m}$, where $m$ is the number of machines. 
\end{theorem}

\begin{proof}
    As mentioned earlier, a configuration is referred to as having divisible speeds if for every pair of machines $i, i^{'} \in M$, where $s_i \geq s_{i^{'}}$, the remainder of the division $\frac{s_i}{s_{i^{'}}}$ is zero. We will show that the following dual fitting constitutes a feasible solution for $Dual_{sum}$, valid for all possible values of $s$, $\tau^*$, and $\tau$, and thus yields an upper bound of $2-\frac{1}{2\cdot m}$ for the price of anarchy. Here $\mathcal{SUB}$ and $CR_i$ are defined as in the proof of Theorem~\ref{thm4}.
\begin{alignat*}{2}
    \beta &= 2 - \frac{1}{2 \cdot m}, \\
    y &= 1, \\
    z_j &=
    \begin{cases}
        i \cdot \frac{1}{2\cdot s_{r_j} \cdot m} & \text{, if }j\in\mathcal{SUB} \text{ and } CR_i \leq j < CR_{i+1} \text{ for some } i \in \{1, \ldots, m-1\}\\
        0  & \text{, otherwise } 
    \end{cases}
\end{alignat*}
 Similar to the proof of Lemma $\ref{minJobSUB}$, let $F$ be the set of all jobs $k$ in the fastest machine that job $k-1$ belongs to a sub-chain. Building upon our previous analysis and considering our proposed solution, we find that for jobs $k \in J$ where $k < CR_1$, for jobs $k$ with $CR_1 < k < CR_m$, $\tau^*_k \neq m$, and $k\neq CR_i$ for $i=2,\ldots,m-1$, and for jobs $k$ in the fastest machine with $k \notin F$, a bound of $\beta \geq 3/2$ is established, thereby satisfying inequality (\ref{eq:GIN}). For $k \in CR$ with $k \neq CR_m$, by using inequality (\ref{eq:IEQopt}), $s_{\tau^*_k}\geq s_{r_k}$, and the fact that $\psi_{\tau^*_k}(\tau^*,k)=0$ we have that
\begin{alignat*}{2}
\beta= 2-\frac{1}{2\cdot m} &\geq 
2 + (z_{k-1}- z_k)\cdot s_{\tau^*_k} \\
&\geq 
\frac{1}{\psi_{\tau^*_k}(\tau^*,k) + 1} + \frac{\psi_{\tau_k}(\tau^*,k)}{\psi_{\tau^*_k}(\tau^*,k) + 1}\cdot \frac{s_{\tau^*_k}}{s_{\tau_k}}+(z_{k-1}- z_k)\cdot s_{\tau^*_k},
\end{alignat*}
where the first inequality is based on our solution, $z_{k-1}- z_k=-\frac{1}{2 \cdot m}$. For $k \in F$, employing the identical analysis as presented in the preceding theorem, we arrive at the following
\begin{alignat*}{2}
\beta= 2-\frac{1}{2\cdot m} &= \frac{3}{2} + \frac{m-1}{2\cdot s_{r_k}\cdot m}\cdot s_{r_k}\\ 
& \geq \frac{3}{2} + z_{k-1} \cdot s_{r_k} = \frac{3}{2} +(z_{k-1}- z_{k})\cdot \frac{s_{\tau^*_{k}}}{\psi_{\tau^*_{k}}(\tau^*,{k})+1},
\end{alignat*}
The initial inequality stems from the observation that a maximum of $(m-1)$ critical jobs can influence the increase in the value of the dual variable related to the SPT constraints and $r_k = r_{k-1}$ since they both are in the same sub-chain. The third equality is grounded in the assumption of machines' divisible speeds, resulting in $\psi_{\tau^*_{k}}(\tau^*,{k})+1 = \frac{s_{\tau^*_{k}}}{s_{r_k}}$. Hence, inequality (\ref{eq:GIN}) is satisfied for $k \in F$. For the remaining jobs in the fastest machine, since the dual variables related to the SPT constraints are zero,  we obtain the inequality of $\beta \geq \frac{3}{2}$, which is satisfied. We conclude that the introduced dual fitting is a feasible solution for $Dual_{sum}$ that is valid for all possible values of $\tau^*$, $\tau$, and $s$.\qed
\end{proof}

\section{Conclusion and Future Research Ideas}

In this study, we tackled a simple yet challenging problem of related machine scheduling games. We made our contribution by presenting an improved upper bound of $2-1/(2\cdot(2m-1))$ for the price of anarchy, which is strictly below 2 for any fixed number of machines. This bound is also valid as an approximation factor for the Ibarra-Kim algorithm. Our analysis was motivated by the gap between the lower bound example with a price of anarchy of $\frac{e}{e-1}$ and an upper bound of $2$ as shown by Hoeksma and Uetz \cite{hoeksma2019price}. We extended the work of Brokkelkamp \cite{closePHD}, who demonstrated that using the weak LP the bound cannot be lower than 5/3. We conjecture that a further improvement can be achieved by identifying a dual fitting for the stronger version of the LP in \cite{closePHD} that considers all possible deviations of a job. Unfortunately, the summation technique yields weaker bounds for this version of the LP and this makes it hard to find a dual fitting.
%
%
%
\bibliographystyle{splncs04}

\end{document}